\documentclass[letterpaper,11pt]{article}
\usepackage{fullpage}

\usepackage{amsmath, amssymb, amsthm}

\newcommand{\E}{\mathbb{E}}
\newcommand{\R}{\mathbb{R}}
\newcommand{\reals}{\mathbb{R}}

\newtheorem{theorem}{Theorem}[section]

\newtheorem{remark}{Remark}[section]
\newtheorem{lemma}[theorem]{Lemma}

\newtheorem{claim}[theorem]{Claim}

\theoremstyle{definition}
\newtheorem{definition}{Definition}[section]
\newtheorem{example}{Example}[section]

\begin{document}

\title{Conducting Truthful Surveys, Cheaply}
\author{Aaron Roth\thanks{Department of Computer and Information Sciences, University of Pennsylvania. Email: {\tt aaroth@cis.upenn.edu}} \and Grant Schoenebeck\thanks{Department of Computer Science, Princeton University. Email: {\tt gschoene@cs.princeton.edu}}}

\maketitle

\begin{abstract}
We consider the problem of conducting a survey with the goal of obtaining an unbiased estimator of some population statistic when individuals have unknown costs (drawn from a known prior) for participating in the survey. Individuals must be compensated for their participation and are strategic agents, and so the payment scheme must incentivize truthful behavior. We derive optimal truthful mechanisms for this problem for the two goals of minimizing the variance of the estimator given a fixed budget, and minimizing the expected cost of the survey given a fixed variance goal.
\end{abstract}

\section{Introduction}
Consider the idealized process of conducting a survey: individuals $x$ can be randomly sampled from some population $X$, and can be asked a question $q$ for some $q:X\rightarrow \{0,1\}$. If the goal is to estimate $\E_{x \in X}[q(x)]$, then the empirical average of $q(x)$ gives us an unbiased estimator, and indeed, only $O(1/\alpha^2)$ samples are required to estimate this quantity to within accuracy $\pm \alpha$ with high probability.

However, as with many things in life, conditions are rarely ideal. Although it may be possible to randomly sample individuals from a population, it is not always possible to convince them to participate in a survey.  Participation is not without cost, and individuals are self interested. Therefore, although the population sampled for the survey may be random, the population that actually responds to the survey is biased towards the subset of the population that has a low cost for participating, which may significantly distort the results. This problem can be partially alleviated by compensating individuals for their participation. But what is the right level of compensation? For a variety of reasons, including varying time constraints and privacy concerns, the cost for participating in a survey can vary substantially from individual to individual. Any fixed level of compensation\footnote{Any fixed level of compensation below some \emph{population wide maximum} on costs: but we want to avoid having to pay each individual \$1 million to respond to our surveys.}  will therefore still result in a biased sample.  We assume that either the agents will not lie about the \emph{content} of the survey, or that their responses are somehow otherwise verifiable.

On the other hand, costs for participation are known only to the individuals and not to the surveyor, and so it is not possible to simply pay each sampled individual exactly their cost for participation. They cannot simply be asked their cost, because as rational agents, they can mis-report their true costs if it is financially beneficial for them to do so. The surveyors problem therefore lies in the realm of mechanism design: he must determine a payment rule that properly incentivizes individuals from the population to report their true cost for participation.

In this paper, we initiate the formal study of the surveyors problem from the perspective of mechanism design. We consider the problem of computing an unbiased estimator of some statistic $\E_{x \in X}[q(x)]$ when individuals from $X$ have varying costs for participation drawn from a known prior $\mathcal{F}$. We derive optimal truthful mechanisms for minimizing the variance of the estimate given a fixed budget on the expected cost of the survey, and for minimizing the expected cost of the survey given a fixed constraint on the variance of the estimator. Our mechanism is simple, practical, and ex-post dominant strategy truthful.

\subsection{Our Results}
Given a prior distribution on costs $\mathcal{F}$, we show how to conduct an survey that optimally trades off the expected variance\footnote{instead of optimizing the expected variance, we optimize for a closely related quantity that may differ by at most $1/n$ from the actual variance} of an unbiased estimator with the expected cost. We first prove a characterization theorem that greatly simplifies our design space: we show that for any truthful surveying mechanism, there exists another truthful surveying mechanism that merely makes take-it-or-leave-it offers that constructs an identical estimator, and has the same expected cost. Therefore, we may without loss of generality restrict our attention to take-it-or-leave-it survey mechanisms when seeking to optimize any function of the survey's estimator as it relates to the expected cost of the survey.

Then, under an extremely mild regularity condition on the prior distribution on costs $\mathcal{F}$, we derive the form of a surveying mechanism which optimizes any continuous, convex, and monotone function of the cost and variance of the mechanism. In particular, this allows us to derive the surveying mechanism which achieves the minimum variance unbiased estimator of a population statistic subject to a cost constraint, and similarly, the minimum cost mechanism subject to a variance constraint.

Because we are optimizing over the class of take-it-or-leave-it mechanisms (to be defined), our optimization problem reduces to finding the optimal distribution over offer prices. We note that if this were an \emph{unconstrained} optimization problem, it would end up being a relatively simple exercise in the calculus of variations. However, because we must solve for a function which is a valid probability distribution, we are facing a \emph{constrained} optimization problem which results in a number of technical complications.

\subsection{Related Work}
The problem of survey design has a long history, and is the topic of an entire field of inquiry in statistics. The idea that survey results may be biased because different individuals may have different proclivities for answering the survey as a function of their types (due to stigmatizing questions) goes back at least to Warner \cite{War65}, who proposed a ``randomized response'' approach to lessen the cost of answering a stigmatizing question. The statistical literature tends to approach the problem by making distributional assumptions about the relationship between survey response and data type, and attempts to avoid bias by compensating for these differences or imputing the values of missing data. We do not attempt to survey this vast literature, but refer the interested reader to useful texts on this topic \cite{Surveybook1,Surveybook2}. In particular, our work differs from this literature in that we do not make any assumptions about the relationship between data value and response rate, but instead take a mechanism-design approach and model survey respondents as rational agents with private costs for participating in the survey, and insist on \emph{truthful} survey mechanisms which incentivize sufficient response rate to provide accurate statistical estimates, without assumptions on how costs are related to types. Crucially, we are prepared to \emph{pay} individuals for their participation, and rely on their rational responses to guarantee the accuracy of our estimators.

The field of Bayesian optimal mechanism design originated with the groundbreaking work of Myerson \cite{Mye81}, which (in certain settings) characterized the auction rule which maximizes seller revenue, given a known prior on the distribution over buyer values. See, e.g., \cite{Auctionbook} for a survey of the literature that has followed. In this work, we take the methodology of Bayesian optimal mechanism design and apply it to the problem of survey design. We assume a known prior over \emph{costs} for individuals participation, but assume that the correlation between costs and survey responses is worst-case. We then design truthful survey mechanisms for calculating unbiased estimators while minimizing various objective functions (i.e. minimizing the variance of the estimator given a cost constraint, or minimizing the cost of the survey given a variance constraint, as well as more general objectives).

Recently, the problem of designing truthful mechanisms to estimate statistics from a population which explicitly experiences costs for privacy loss was introduced by Ghosh and Roth \cite{GR11}. Subsequently (and concurrently with this work), Ligett and Roth \cite{LR12} extend this work to sequences of take-it-or-leave-it offers. Although it has similar goals, this line of work differs from the current paper in that it measures cost using the formalism of differential privacy, and more importantly, takes a worst-case view and does not assume a known prior over agent costs. In contrast, here we use a prior over agent costs to derive \emph{optimal} mechanisms. This paper can be viewed as answering an open question posed by \cite{GR11}, which asked whether the approach of Bayesian-optimal mechanism design could be brought to bear on the data gathering problem when the distribution over agent costs was known.

\section{Preliminaries}
\subsection{Model and Mechanism Design Basics}
We consider a surveyor who is interested in some population statistic defined by a predicate on individuals, $q:X\rightarrow \{0,1\}$. We model each individual $x_i \in X$ as having an unknown cost $c_i \geq 0$ for participating in the survey (i.e. for revealing the bit $q(x_i)$). Individuals and costs $(x_i, c_i)$ are drawn i.i.d. from a distribution $\mathcal{D}$: $(x_i, c_i) \sim \mathcal{D}$, and the surveyor wishes to estimate $\E_{(x_i, c_i) \sim \mathcal{D}}[q(x_i)]$. Note that $c_i$ and $x_i$ can be arbitrarily correlated. We denote by $\mathcal{F}$ the marginal distribution of $\mathcal{D}$ on $c$.
We write $f$ for the probability density function for $\mathcal{F}$, and write $F$ for the cumulative density function.  We assume that $f$ is piecewise differentiable except, possible, on a measure 0 set.  The surveying mechanism knows (and can be a function of) $\mathcal{F}$, the marginal distribution on costs, but has no other knowledge of $\mathcal{D}$. We write $n$ to denote the number of individuals sampled by the surveying mechanism.

A survey mechanism consists of an allocation rule $A:\mathbb{R}^+\rightarrow [0, 1]$ and a possibly randomized payment rule $\hat P:\mathbb{R}^+\rightarrow \mathbb{R}_{\geq 0}$. We define $P:\mathbb{R}^+\rightarrow \mathbb{R}_{\geq 0}$ to be $P(c) = \E[\hat P(c)]$.
At time $i$, the mechanism can sample an individual $(x_i, c_i) \sim \mathcal{D}$, elicit their reported cost $\hat{c}_i$, and then with probability $A(\hat{c}_i)$ collect the survey sample $q(x_i)$ in exchange for a payment $\hat P(\hat{c}_i)$. With the remaining probability $1-A(\hat{c}_i)$ the mechanism does not conduct the survey and does not make any payment.


We model individuals as having quasi-linear utility functions and being expectation maximizers.  An individual with cost $c_i$ who reports cost $\hat{c}_i$ to a surveying mechanism, experiences cost $P(\hat{c}_i)-c_i$ when the survey is conducted, which occurs with probability $A(\hat{c}_i)$. Therefore their expected utility is:
$$u(c_i, A(\hat{c}_i), P(\hat{c}_i)) = A(\hat{c}_i)\cdot \left(P(\hat{c}_i) - c_i  \right)$$
Note that this only depends on  $\hat P$ via its expectation $P$.
We want our surveying mechanisms to be able to give accuracy guarantees even in the presence of rational, selfish agents who wish to maximize their own utility. Therefore we require that our mechanisms be truthful and individually rational.
\begin{definition}[Truthfulness]
A surveying mechanism $(A, \hat P)$ is truthful if for all $c_i, \hat{c}_i \in \reals^+$:
$$u(c_i, A(c_i), \hat P(c_i)) \geq u(c_i, A(\hat{c}_i), \hat  P(\hat{c}_i))$$
\end{definition}
\begin{definition}[Individual Rationality]
A surveying mechanism  $(A,P)$ is individually rational if for all $\hat{c}_i$ resulting in the survey being conducted:
$$P(\hat{c}_i) \geq \hat{c}_i$$
\end{definition}
Informally, truthfulness states that a utility-maximizing agent can do no better than reporting his true cost faithfully to the surveying mechanism, and individual rationality states that no individual can ever be forced to experience negative utility.

Truthfulness and individual rationality impose the constraints on our design space for truthful surveys. We now define the objective that we wish to maximize within this space. When a transaction is made with an individual $i$ through the mechanism (which occurs with probability $A(c_i)$), he supplies his bit $q(x_i) \in \{0,1\}$. The mechanism compiles these bits into an \emph{estimator} $S$ of the statistic $\E_{(x,c) \sim \mathcal{D}}[q(x)]$. We require that our mechanisms give unbiased estimates.
\begin{definition}[unbiased estimator]
An estimator $S = S(\mathcal{D}, A, n)$ of the statistic $\E_{(x, c) \sim \mathcal{D}}[q(x)]$ is \emph{unbiased} if:
$$\E[S(\mathcal{D}, A, n)] = \E_{(x, c) \sim \mathcal{D}}[q(x)]$$
where the randomness of $\E[S(\mathcal{D}, A, n)]$ is taken over the random choices of the sampled individuals $(x_i, c_i) \sim \mathcal{D}$ for $i \in [n]$ and the internal randomness of the mechanism allocation rule $A$.
\end{definition}

We wish to design mechanisms for constructing unbiased estimators of $\E_{(x, c) \sim \mathcal{D}}[q(x)]$ while minimizing \emph{variance}.

\begin{definition}[variance]
The \emph{variance} of a random variable $S(\mathcal{D}, A, n)$ is:
$$\textrm{Var}(S(\mathcal{D}, A, n)) = \E[(S(\mathcal{D}, A, n) - \E[S(\mathcal{D}, A, n)])^2]$$
\end{definition}

Note that $\textrm{Var}(S)$ depends on the unknown distribution $\mathcal{D}$\footnote{For example, most reasonable estimators $S(\mathcal{D}, A, n)$ will have $\textrm{Var}(S) = 0$ if $\mathcal{D}$ is such that $\Pr_{(x, c) \sim \mathcal{D}}[q(x) = 1] = 0$.}, even though the mechanism knows only the marginal distribution over costs, $\mathcal{F}$. Therefore, our measure of performance will be worst case variance over all distributions $\mathcal{D}$ consistent with the marginal distribution $\mathcal{F}$:

\begin{definition}[Worst case variance]
The worst case variance of a random variable $S(\mathcal{D}, A, n)$ given a marginal distribution on costs $\mathcal{F}$ is:
$$\textrm{Var}^*(S) = \max_{\mathcal{D}}\textrm{Var}(S(\mathcal{D}, A, n))$$
where the maximum is taken over all distributions $\mathcal{D}$ consistent with $\mathcal{F}$.
\end{definition}

When the distribution is not known, it is not always possible to design the minimum variance unbiased estimator of a distribution parameter.  Therefore, we will concentrate on \emph{linear estimators} in this paper.

\begin{definition}
An estimator $S$ is \emph{linear} if for each individual $i$ from which the mechanism elicits $q(x_i)$, there is a multiplier $\beta(c_i)$ such that:
$$S = \sum \beta(c_i)q(x_i)$$
\end{definition}

We will want to measure the expected cost of running a truthful survey. The mechanism only needs to make a payment $P(c_i)$ when a transaction with agent $i$ occurs, which happens with probability $A(c_i)$. Therefore we can derive the expected cost of a truthful mechanism with prior $\mathcal{F}$. The following claim is immediate:
\begin{claim}
The expected cost of surveying $n$ individuals with a truthful mechanism $(A, P)$ is:
$$\E[\textrm{cost}(A, P, n)] = n\cdot \int_0^\infty A(c)\cdot P(c)\cdot f(c)\ dc$$
\end{claim}

For intuition's sake, we here give an example of a non-trivial truthful lottery mechanism:
\begin{example}
For values $c_i \in [1, \infty]$, consider a lottery mechanism with $A(x) = 1/x^2$ and $P(x) = 2x$. Note that we always have $P(x) \geq x$ and so the mechanism is individually rational. Observe that it is also truthful. For any individual $i$, we have:
$$u_i(c_i, A(\hat{c}_i), P(\hat{c}_i)) = \frac{1}{\hat{c}_i^2}\left(2\hat{c}_i - c_i\right) = \frac{2}{\hat{c}_i} - \frac{c_i}{\hat{c}_i^2}$$
taking the derivative with respect to $\hat{c}_i$, we can see that this quantity is always maximized when $\hat{c}_i = c_i$, which is the condition needed for truthfulness.
\end{example}

\subsection{Calculus of Variations Basics} \label{subsec calc of var}
In this paper we will seek to minimize the cost and variance of our mechanisms, which are functionals of our chosen allocation rule $A$ and pricing rule $P$. Functional minimization problems are addressed by the Calculus of Variations. Here we give an informal introduction to the required (very basic) preliminaries.  This can be found, for example, in~\cite{Liberzon}.

Let $U$ be a space of functions over which we would like to optimize.
Let $M: U \rightarrow \R$.
The first variation (which is corresponds to the Gateaux derivative) of $M$ at $G \in U$ in the direction $\hat G$ is simply

$$\delta M_{|G}(\hat G) =\lim_{\epsilon \rightarrow 0} \left[\frac{M((G + \epsilon \hat G) - M(G)}{\epsilon} \right]$$

A direction is \emph{feasible} if $G + \epsilon \hat G \in U$ for sufficiently small $\epsilon$.  We only deal with convex $U$.  In this case every \emph{feasible direction} of a derivative at $G$ can be written as $\hat G - G$ for some $\hat G \in U$.

For a convex $U$, we say that $G$ is a minimizer (or local minimum) of $M$ if for every feasible direction $\hat G$ we have that $\delta M_{|G}(\hat G) \geq 0$.

For a convex $U$, we call $\delta^2 M_{|G}: U \rightarrow \R$ the second variation of $M$ at $G$ if
$$M(G + \epsilon \hat G) = M(G) + \epsilon \delta M_{|G}(\hat G) + \epsilon^2\delta^2 M_{|G}(\hat G) + o(\epsilon^2)$$ for all feasible directions $\hat G \in U$ and all $\epsilon > 0$

$M$ is \emph{convex} if $\delta^2 M_{|G} \geq 0$ for all $G \in U$.  $M$ is \emph{strictly convex} if $\delta^2 M_{|G} > 0$ for all $G \in U$.  If $M$ and $U$ are convex than any minimizer is also a global minimum.   If $M$ is strictly convex and $U$ is convex than any minimizer is also a unique global minimum.

We will use Lagrange Multipliers.  If we are trying to minimize $M: U \times \R \rightarrow \R$ subject to the constraint $H:U \times \R \rightarrow \R$,  then by properties of Lagrange Multipliers we must have that any direction that decreases $M$ is in the tangent space of the constraints.  Thus it is sufficient to show
$\partial_G [M(G^*, n) + \lambda  H(G^*, n)] \geq 0$, $\partial_{\lambda} [M(G^*, n) + \lambda  H(G^*, n)] = 0$, and $\partial_n [M(G^*, n) + \lambda  H(G^*, n)] = 0$ for some $\lambda \in \R$.

\section{Simplifying the Design Space}
In this section, we show that without loss of generality, we can restrict our attention to a simple class of mechanisms.
\begin{definition}
A \emph{take-it-or-leave-it} (TIOLI) lottery mechanism is defined by a distribution $\mathcal{G}$ over $\mathbb{R}^+$ with probability density function $g$ and cumulative density function $G$. Given a price $x$ from an agent, $A(x)$ and $P(x)$ are then computed as follows:
\begin{enumerate}
\item Sample $p_i \sim \mathcal{G}$.
\item Set $A(x) = 0$ if $x > p_i$ and $A(x) = 1$ if $x \leq p_i$.
\item Set $\hat{P}(x) = 0$ if $x > p_i$ and $\hat{P}(x) = p_i$ if $x \leq p_i$.
\end{enumerate}
\end{definition}

Thus $A$ and $\hat P$ are completely determined by $\mathcal{G}$.  Namely $A(c) = 1 - G(c)$; and $\hat P(c)$ is the distribution $\mathcal{G}$ conditioned on being greater than $c$.

We first observe that a TIOLI lottery mechanism is truthful and individually rational.
\begin{theorem}
Any take-it-or-leave-it lottery mechanism is truthful and individually rational.
\end{theorem}
\begin{proof}
It is immediate that TIOLI mechanisms are individually rational, since whenever $A(c_i) = 1$ we have $c_i \leq p_i$ and so: $P(c_i) = p_i \geq c_i$. To see that the mechanism is truthful, observe that for each fixed $p_i$, if $c_i \leq p_i$ then individual $i$ gets utility $p_i - c_i$ for every reported valuation $\hat{c}_i \geq p_i$ and utility $0$ for every reported valuation $\hat{c}_i < p_i$. If $c_i \leq p_i$ then $p_i - c_i \geq 0$ and the so the utility of reporting $\hat{c}_i = c_i$ is non-negative and is not changed by over-reporting $\hat{c}_i > c_i$ (and under-reporting can only reduce the payoff to 0). If $c_i >p_i$ then $p_i - c_i < 0$, which cannot be improved by under-reporting $\hat{c}_i < c_i$ (and over-reporting can only result in negative utility).
\end{proof}
We now show that without loss of generality, we may restrict our attention to TIOLI lottery mechanisms: that is, every lottery mechanism $(A, P)$ is implementable as a TIOLI lottery mechanism with the same allocation rule and the same expected payment.
\begin{lemma}
\label{lem:monotone}
If $(A, P)$ is a truthful lottery mechanism, then $A(x)$ must be monotonically decreasing in $x$.
\end{lemma}
\begin{proof}
Let $x_1, x_2 \in \mathbb{R}^+$. We will show that if $A(x_2) > A(x_1)$ then $x_1 > x_2$. By the truthfulness of the mechanism $(A, P)$ we have that for all $c$:
$$u_i(c, A(x), P(x)) = A(x)(P(x) - c) = A(x)P(x) - A(x)c$$
is maximized at $x = c$. In particular, considering $c = x_1$ and $c = x_2$ we have the following pair of inequalities:
$$A(x_1)P(x_1) - A(x_1)x_1 \geq A(x_2)P(x_2) - A(x_2)x_1$$
$$A(x_2)P(x_2) - A(x_2)x_2 \geq A(x_1)P(x_1) - A(x_1)x_2$$
Adding these two inequalities, we find:
$$x_1(A(x_2)-A(x_1)) + x_2(A(x_1) - A(x_2)) \geq 0$$
Note that by assumption, $(A(x_2) - A(x_1)) > 0$. Therefore, dividing both sides by $(A(x_2) - A(x_1))$ we have:
$$x_1 - x_2 \geq 0$$
which completes the proof.
\end{proof}

\begin{theorem}
Every truthful differentiable allocation rule $A:\mathbb{R}^+\rightarrow [0,1]$ is implementable as a take-it-or-leave-it lottery mechanism with the same allocation rule $A(x)$ and the same expected payment $\E[A(x)P(x)]$.
\end{theorem}
\begin{proof}
Fix an allocation rule $A^*$ and a corresponding payment rule $P^*$ so that $(A^*, P^*)$ is a truthful lottery mechanism. Define a TIOLI lottery mechanism with allocation and payment rules denoted $A, P$ respectively, by defining the distribution over offers $\mathcal{G}$ to be the distribution with probability density function $g(x) = -\frac{d}{dx}A^*(x)$. Note that by lemma \ref{lem:monotone} we have $g(x) \geq 0$ for all $x$, and $\int_0^\infty g(x) = A^*(0) \leq 1$, and so $g$ is a valid probability density function. Now consider the probability that this TIOLI mechanism transacts with an individual $i$ with cost $c_i$: this occurs exactly when the sampled offer $p_i$ is such that $p_i \geq c_i$, which occurs with probability:
$$A(c_i) = \int_{c_i}^\infty g(x) d x = -\int_{c_i}^\infty  \frac{d}{dx}A^*(x)  = A^*(c_i)$$
Thus the given TIOLI lottery mechanism implements exactly the allocation rule $A$. It remains to show that the expected payment of both mechanisms are the same:  $\E[A(x)P(x)] = A^*(x)P^*(x)$ for all $x$.
\begin{eqnarray*}
\E[A(x)P(x)] &=&  \E[P(x) | p_i \geq x]\cdot \Pr[p_i \geq x] \\
&=& \left(\frac{\int_{x}^\infty t \cdot g(t) dt}{A(x)}\right)\cdot A(x) \\
&=&  \int_{x}^\infty t \cdot g(t) dt
\end{eqnarray*}
Now consider the utility of agent $i$ reporting value $x$ to mechanism $(A^*, P^*)$:
$$u(c_i, A^*(x), P^*(x)) = A^*(x)\left(P^*(x) - c_i\right) = A^*(x)P^*(x) - c_iA^*(x))$$
By the truthfulness of this mechanism, for every value of $c_i$, this expression is maximized by setting $x = c_i$. Write the function $A^*P^*(x) = A^*(x)P^*(x)$. Taking the derivative of the above expression and setting it to zero, we have simultaneously for all values $c$:
$$\frac{d}{d c} (A^*P^*(c)) = c\cdot \frac{d}{d c}(A^*(c))$$
Integrating both sides of the above equality from $x$ to $\infty$ we have:
\begin{eqnarray*}
-A^*(x)P^*(x)  &=& \int_{x}^\infty \frac{d}{d c} (A^*P^*(c)) dc \\
&=& \int_{x}^\infty c\cdot \frac{d}{d c}(A^*(c))dc \\
&=& -\int_{x}^\infty c\cdot g(c) dc
\end{eqnarray*}
In other words, for all $x$:
$$A^*(x)P^*(x) = \int_{x}^\infty c\cdot g(c) dc = \E[A(x)P(x)]$$
\end{proof}
\begin{remark}
The requirement that the allocation rule is differentiable is for technical convenience, and can be relaxed to a more mild piecewise differentiable condition.
\end{remark}
We have just seen that with respect to allocation and payment rules, we can restrict our attention without loss of generality to TIOLI mechanisms. What about unbiased estimators? In our setting, as it happens, the celebrated Horvitz-Thompson estimator is the \emph{unique} unbiased \emph{linear} estimator, and therefore the linear unbiased estimator of minimum variance, independent of the distribution $\mathcal{D}$.
\begin{definition}[Horvitz-Thompson Estimator]
For an allocation rule $A(x)$, the Horvitz-Thompson estimator computes:
$$S_{HT}(\mathcal{D}, A, n) = \frac{1}{n}\sum_{i=1}^n d_i$$
where:
$$d_i = \left\{
          \begin{array}{ll}
            \frac{q(x_i)}{A(c_i)}, & \hbox{If the mechanism transacts with individual $i$;} \\
            0, & \hbox{Otherwise.}
          \end{array}
        \right.$$
\end{definition}

\begin{theorem}[\cite{HT52}]
For every distribution $\mathcal{D}$ and every allocation rule $A$, $S_{HT}(\mathcal{D}, A, n)$ is the minimum variance linear unbiased estimator of $\E[q(x_i)]$.
\end{theorem}
Moreover, the Horvitz-Thompson estimator is \emph{admissible} among the set of all unbiased estimators of $\E_{(x, c) \sim \mathcal{D}}[q(x)]$ \cite{GJ65} (even non-linear estimators). This means that the estimator is \emph{un-dominated}: there is no other estimator that has lower variance than $S_{HT}$ with respect to every distribution $\mathcal{D}$.

Therefore, for the rest of the paper, we restrict our attention to TIOLI mechanisms that use the Horvitz-Thompson estimator.

\subsection{Price and Variance of TIOLI mechanisms that use the Horvitz-Thompson estimator}

In this section we derive the cost and variance of the TIOLI mechanisms that use the Horvitz-Thompson estimator.

Recall that for a TIOLI mechanism with price distribution $\mathcal{G}$, we have: \begin{eqnarray*}
\E[\textrm{Cost}(A(\mathcal{G}), P(\mathcal{G}),n)] &=& n\cdot \int_0^\infty A(c)\cdot P(c)\cdot f(c) dc \\
&=&  n\cdot \int_0^\infty f(c)\cdot\left(\int_c^\infty x\cdot g(x) dx\right) \ dc \\
&=& n \cdot \int_0^\infty c\cdot g(c)\cdot F(c) dc
\end{eqnarray*}

We can compute (up to additive error $1/n$) the worst case \emph{variance} of a TIOLI mechanism characterized by distribution  $\mathcal{G}$:
\begin{lemma}
The worst-case variance of a Horvitz-Thompson TIOLI mechanism parameterized with distribution $\mathcal{G}$ used on a marginal value distribution $\mathcal{F}$ for $n$ iterations is:
$$\frac{1}{n}\int_{1}^\infty\frac{f(x)}{1-G(x)} dx - \frac{1}{n}\leq \textrm{Var}^*(A(\mathcal{G}), \mathcal{F},n) \leq \frac{1}{n}\int_{1}^\infty\frac{f(x)}{1-G(x)} dx$$
\end{lemma}
\begin{proof}
Recall that given an allocation rule $A$, the Horvitz-Thompson estimator is defined as:
$$S_{HT}(\mathcal{D}, A, n) = \frac{1}{n}\sum_{i=1}^n d_i\ \  \textrm{where}\ \ d_i = \left\{
          \begin{array}{ll}
            \frac{q(x_i)}{A(c_i)}, & \hbox{If the mechanism transacts with individual $i$;} \\
            0, & \hbox{Otherwise.}
          \end{array}
        \right.$$

We will analyze $\textrm{Var}(d_1)$ and then observe that $\textrm{Var}(\frac{1}{n}\sum_{i=1}^n d_i) = \frac{1}{n}\textrm{Var}(d_1)$ since each $d_i$ is independent and identically distributed. First note that $\textrm{Var}(d_1) = \E[d_1^2] - (E[d_1])^2$. We observe that $0 \leq E[d_1]^2 \leq 1$ and that:
\begin{eqnarray*}
\E[d_1^2] &=& \int_0^\infty \left(f(v)\cdot A(v)\cdot \Pr[q(x) = 1 | c = v]\cdot \left(\frac{1}{A(v)}\right)^2\right) dv \\
&\leq& \int_0^\infty \left(f(v)\cdot A(v)\cdot  \left(\frac{1}{A(v)}\right)^2\right) dv \\
&=& \int_0^\infty \frac{f(v)}{A(v)} dv \\
&=& \int_0^\infty \frac{f(v)}{1-G(v)} dv
\end{eqnarray*}
with equality holding for the distribution that sets $q(x) = 1$ independent of the valuation $c$. Hence, we have:
$$\int_0^\infty \frac{f(v)}{1-G(v)} dv - 1 \leq \textrm{Var}^*(d_1) \leq  \int_0^\infty \frac{f(v)}{1-G(v)} dv$$
as claimed.
\end{proof}

Let $\textrm{V}^*(A(\mathcal{G}), \mathcal{F}, n) = \int_0^\infty \frac{f(v)}{1-G(v)} dv$.  Instead of optimizing for the worst-case expected variance $\textrm{Var}^*(A(\mathcal{G}), \mathcal{F}, n)$ we will optimize for $\textrm{V}^*(A(\mathcal{G}), \mathcal{F}, n)$, since $\textrm{V}^*(A(\mathcal{G}), \mathcal{F}, n)$ is a very good approximation and will be much more  convenient.
Additionally, we note that if we change the estimator slightly to be
$$\tilde{S}_{HT}(\mathcal{D}, A, n) = \frac{1}{n}\sum_{i=1}^n d_i\ \  \textrm{where}\ \ d_i = \left\{
          \begin{array}{ll}
            \frac{(-1)^{q(x_i)}}{A(c_i)}, & \hbox{If the mechanism transacts with individual $i$;} \\
            0, & \hbox{Otherwise.}
          \end{array}
        \right.$$

Then this is an unbiased estimator of $1- 2\E[q(x_i)]$ with expected variance exactly $V^*(A(\mathcal{G}), \mathcal{F}, n)$.

For a TIOLI mechanism $\textrm{Cost}(A(\mathcal{G}), P(\mathcal{G}),n)$ and $V^*(A(\mathcal{G}), \mathcal{F}, n)$ can we written in terms of $\mathcal{F}$, $\mathcal{G}$ and $n$, and no longer need to be stated in terms of $\mathcal{D}$, $A$, or $P$.  Thus we will, in the future, denote them by  $Cost(\mathcal{F}, \mathcal{G}, n)$ and $V^*(\mathcal{F}, \mathcal{G}, n)$, and will often put fixed terms in the subscripts (e.g. $\textrm{Cost}_{\mathcal{F},n}(\mathcal{G})$).

\begin{claim} \label{claim derivative of cost} Let $\mathcal{F}$, $\mathcal{G}$, and $\mathcal{\hat G}$ have piecewise twice differentiable CDFs, then  $$\delta Cost_{\mathcal{F}, n | \mathcal{\hat{G} - G}}(\hat{\mathcal{G}}-\mathcal{G}) =  \int_0^{\infty} (\hat g(y) - g(y)) y F(y) dy$$
\end{claim}

\begin{proof}  Calculations show
\begin{align*}\lim_{\epsilon \rightarrow 0} \left[\frac{Cost(\mathcal{F}, (1 - \epsilon)\mathcal{G} + \epsilon\hat{\mathcal{G}}, n) -  Cost(\mathcal{F}, \mathcal{G}, n)}{\epsilon}\right] & = n \int_0^{\infty} (\hat{g}(y) - g(y))[ y F(y)] dy \end{align*}
\end{proof}
\begin{claim} \label{claim derivative of variance}  Let $\mathcal{F}$, $\mathcal{G}$, and $\mathcal{\hat G}$ have piecewise twice differentiable CDFs, the $$\delta V^*_{\mathcal{F}, n | \mathcal{G}} (\mathcal{\hat  G}) = \int_0^{\infty} (\hat g(y) - g(y)) \int_0^{x} \frac{f(y)}{(\int_y^{\infty} g(z) dz)^2}dy dx$$
\end{claim}

\begin{proof}  Calculations show
\begin{align*}\lim_{\epsilon \rightarrow 0} \left[\frac{Var(\mathcal{F}, (1 - \epsilon)\mathcal{G} + \epsilon \hat{\mathcal{G}}, n) -  Var(\mathcal{F}, \mathcal{G}, n)}{\epsilon}\right] & = \frac{1}{n} \int_0^{\infty} \frac{f(x) \int_x^{\infty} \hat{g}(y) - g(y) dy}{(\int_x^{\infty} g(z) dz)^2}dx \\ &= \frac{1}{n} \int_0^{\infty}(\hat{g}(y) - g(y)) \int_0^y \frac{f(x)}{(\int_x^{\infty} g(z) dz)^2}dx dy \end{align*}  This last step we switch the integrals, which we can do because $f$ and $g$ are piecewise continuous.  The claim follows.
\end{proof}

\section{Deriving Optimal Mechanisms}
In this section, we derive optimal TIOLI mechanisms that use the Horvitz-Thompson estimator. Recall that such mechanisms are defined by a distribution $\mathcal{G}$ over \emph{prices} $p \in [0, \infty)$ with CDF $G$ and pdf $g$.  In the previous section, we derived formulas for the cost, $Cost(\mathcal{F}, \mathcal{G}, n)$, and worst-case expected variance,  $V^*(\mathcal{F}, \mathcal{G}, n)$ with error at most $ \pm 1/n$.

We now derive optimal mechanisms for minimizing the variance subject to a cost constraint, and minimizing cost subject to a constraint on the variance, where $\mathcal{F}$ is a distribution that satisfies a certain regularity condition.  In the case where the number of participants is fixed, we can extend this to design
optimal mechanisms for any possible convex minimization objective function.

We will assume that $F$ is twice differentiable for all positive $x$, that $n$ is fixed
We remark that this assumption can be further relaxed so that $f$ is piecewise twice differentiable.

We optimize over the class of distributions with piece-wise twice differentiable CDF.  Also, without loss of generality, we require that $\mathcal{G}$ has no support less than every point in the support of $\mathcal{F}$ except possibly a point mass at zero.  This is clearly without loss of generality as any such offer will never be accepted.  We do this so that we can (often) claim the uniqueness of optimal solutions.

We first define what the optimal will be.  To this end we first define
\begin{align*}
 \tilde{G}_{\mathcal{F}, \alpha}(x) & = 1 - \sqrt{\alpha\frac{f(x)}{F(x) + x f(x)}} \\
 \tilde g(x)_{\mathcal{F}, \alpha} &=  \sqrt{\alpha} \frac{2 f(x)^2 - F(y) f'(y)}{(y f(y) + F(y))^{\frac{3}{2}}f(y)^{\frac{1}{2}}} = \frac{d}{dx}[\tilde{G}_{\mathcal{F}, \alpha}(x)]
\end{align*}

As above, if $G$ is a CDF, we will denote the corresponding density function by $g$.

For the correct setting of $\alpha$, $\tilde G_{\mathcal{F}, \alpha}$ will be a (possibly infeasable) CDF of the local minimum.  $\tilde G_{\mathcal{F}, \alpha}$ may be decreasing or negative and thus infeasible as a cumulative density function.  We first construct a monotone version of $\tilde G_{\mathcal{F}, \alpha}$, then we make it positive.

Next we define $\{ \bar I_k \}$.  Because $F$ is twice differentiable everywhere, $2f(x)^2 \geq f'(x)F(x)$ except on a countable number of intervals where $2f(x)^2 < f'(x)F(x)$.  Let these intervals be $\{( i_k, j_k)\}$.

For each $(i_k, j_k)$ define
\begin{eqnarray*}
K_k(x) & = & \int_{x}^{j_k} \frac{f(y)}{(\int_y^{\infty} \tilde g(z) dz)^2} - \frac{f(y)}{(\int_{\bar j_k}^{\infty} \tilde g(z) dz)^2} dy \\
\check i_k & = & sup_{x \leq i_k} K_k(x) \leq 0.
\end{eqnarray*}
to get the new interval $(\check i_k, j_k)$.  Discard any intervals that are subsets of other intervals, and relabel the remaining intervals as $\bar I_1 = [\bar i_1, \bar j_1), \bar I_2 = [\bar i_2, \bar j_2), \ldots$ where $0 \leq \bar i_1 < \bar j_1 \leq \bar i_{\bar 2} < \bar j_{\bar 2} \leq \ldots$.  Let $\bar I = \cup \bar I_k$.

Let
\begin{eqnarray*}
 \check{G}_{\mathcal{F}, \alpha}(x) & =& \left\{\begin{array}{cc} \tilde G_{\mathcal{F}, \alpha}(x) & x \not\in \bar I \\
                                             \tilde G_{\mathcal{F}, \alpha}(\bar j_k)           & x \in  \bar I_k\end{array} \right.\\
 \bar x_{\mathcal{F}, \alpha} &=&  \max \{0, \inf_x \int_x^{\infty} \tilde{g}_{\mathcal{F}, \alpha}(x) \leq 1\} \\
  \bar G(x)_{\mathcal{F}, \alpha} &=& \max \{\check {G}_{\mathcal{F}, \alpha}(x), 0 \} \\
\end{eqnarray*}

\begin{claim} \label{claim about G} $\bar G(x)_{\mathcal{F}, \alpha}$ is a valid CDF: i.e. for $x \in [0, \infty]$, $\bar G(x)_{\mathcal{F}, \alpha}$ is non-decreasing, $0 \leq \bar G(x)_{\mathcal{F}, \alpha}$, and $\lim_{x \rightarrow \infty} \bar G(x)_{\mathcal{F}, \alpha} = 1$. \end{claim}
\begin{proof} $\tilde g_{\mathcal{F}, \alpha}(x) =  \frac{\sqrt{\alpha}}{n} \frac{2 f(x)^2 - F(x) f'(x)}{(x f(x) + F(x))^{\frac{3}{2}}f(x)^{\frac{1}{2}}},$  which is non-negative whenever $2 f(x)^2 \geq F(x) f'(x)$  because $F(x), f(x), x, n \geq 0$.  Thus, $\check G_{\mathcal{F}, \alpha}(x)$ is non-decreasing for $x \not\in \bar I$ because we have $\check G_{\mathcal{F}, \alpha}(x) = \tilde G_{\mathcal{F}, \alpha}(x)$ and $2 f(x)^2 \geq F(x) f'(x)$.  Also, for  $x \in \bar I$ we have that $\check g_{\mathcal{F}, \alpha}(x) = 0$.  This leaves only $\bar i_k$ and $\bar j_k$ to check.  By inspection $\check G_{\mathcal{F}, \alpha}$ is non-decreasing at $\bar j_k$ for all $k$ (it is continuous and all surrounding derivatives are non-negative).  At any $\bar i_k$ we have that $\check G_{\mathcal{F}, \alpha}(\bar i_k) = \tilde G_{\mathcal{F}, \alpha}(\bar j_k)$, but that at all points to the left  $\check G_{\mathcal{F}, \alpha}(\bar i_k - \epsilon) = \tilde G_{\mathcal{F}, \alpha}(\bar i_k - \epsilon)$,  which limits to $\tilde G_{\mathcal{F}, \alpha}(\bar i_k)$ so we must show that $\tilde G_{\mathcal{F}, \alpha}(\bar i_k) \leq \tilde G_{\mathcal{F}, \alpha}(\bar j_k)$.  However, this follows from the construction of $I_k$ because $K_k(\bar{i}_k) = 0$, but $K_k(\bar{i}_k + \epsilon)  > 0$.  Thus the derivative of $K_k$ at $\bar i_k$ which is $-f(\bar i_k)[(\int_{\bar i_k}^{\infty} \tilde g(z) dz)^{-2} - (\int_{\bar j_k}^{\infty} \tilde g(z) dz)^{-2}] = -f(\bar i_k) [(1 - \tilde G(\bar i_k))^{-2} - (1 - \tilde G(\bar j_k))^{-2}]$ is non-negative, and so $\tilde G_{\mathcal{F}, \alpha}(\bar i_k) \leq \tilde G_{\mathcal{F}, \alpha}(\bar j_k)$.

Because  $\check G$ is non-decreasing, by construction $\bar G$  is also.
Also, by construction $0 \leq G(0)$.  Because $\frac{f(x)}{F(x) + x f(x)} \leq \frac{1}{x}$, for sufficiently large $x$ we have $$ 1-\sqrt{\frac{\alpha}{x}} \leq 1 - \sqrt{\alpha \frac{f(x)}{F + x f(x)}} = \bar G(x) \leq 1$$ and so  $\bar G(x)$ goes to 1 as $x$ increases.
\end{proof}

\begin{remark} \label{remark cont in alpha}  We remark that $\textrm{Cost}(\mathcal{F}, \bar{G}_{\mathcal{F}, \alpha}, n)$ and $\textrm{V}^*(\mathcal{F}, \bar{G}_{\mathcal{F}, \alpha}, n)$ are continuous in $\alpha$ and are monotonically increasing and monotonically decreasing in $\alpha$, respectively. \end{remark}

\begin{theorem} \label{theorem threshold fixed n}  Let $F$ be twice differentiable for all positive $x$, and let $n$ be fixed.  The unique minimizer of $V^*_{\mathcal{F}, n}(G)$ subject to an expected cost constraint $\textrm{Cost}_{\mathcal{F}, n}(\mathcal{G}) \leq  B$,  is of the form $\bar G_{\mathcal{F}, \alpha} (x)$ where $\alpha$ is unique.  Also, the unique minimizer of $\textrm{Cost}_{\mathcal{F}, n}(\mathcal{G})(G)$ subject to a worst-case expected variance constraint $V^*_{\mathcal{F}, n}(G) \leq  B$, is of the form $\bar G_{\mathcal{F}, \alpha} (x)$ where $\alpha$ is unique.
\end{theorem}

\begin{theorem} \label{theorem convex fixed n} Let $\Phi(x_1, x_2)$ be continuous, convex, and monotone in both of its inputs; let $F$ be twice differentiable for all positive $x$, and let $n$ be fixed. The unique minimizer of $\Phi(Cost(\mathcal{F}, \mathcal{G}, n), V^*(\mathcal{F}, \mathcal{G}, n))$ is of the form $\bar G_{\mathcal{F}, \alpha} (x)$ where $\alpha$ is unique.  Moreover, this is also guaranteed to be the global minimum.
\end{theorem}

In this last theorem, the number of participants in the survey is no longer fixed, and we allow a cost to be associated with recruiting each participant independently of whether or not they accept.

\begin{theorem} \label{theorem threshold variable n} Let $F$ be twice differentiable for all positive $x$, and suppose there is a cost $\beta \geq 0$ for recruiting each participant (i.e. independently of whether or not they accept the offer).  If $\beta > 0$, then the unique minimizer of $V^*_{\mathcal{F}}(G, n)$ subject to an expected cost constraint $\textrm{Cost}_{\mathcal{F}}(\mathcal{G}, n) \leq  B$, is of the form $\bar G_{\mathcal{F}, \alpha} (x)$ where $\alpha$ is unique.  If $\beta = 0$, then there is a collection of minimizers of the form $\bar G_{\mathcal{F}, \alpha} (x)$ where $\alpha \propto 1/n$.

Also, if $\beta > 0$, the unique minimizer of $\textrm{Cost}_{\mathcal{F}, n}(\mathcal{G})$ subject to a worst-case expected variance constraint $V^*_{\mathcal{F}, n}(G) \leq  B$,  is of the form $\bar G_{\mathcal{F}, \alpha} (x)$ where $\alpha$ is unique.  If $\beta = 0$, then there is a collection of minimizers of the form $\bar G_{\mathcal{F}, \alpha} (x)$ where $\sqrt{\alpha} \propto 1/n$.
\end{theorem}

\begin{remark}  The requirement that $F$ be twice differentiable everywhere is for technical convenience.  Our theorems can be generalized to the case where the PDF $F$ is piece-wise twice differentiable by first approximating $F$ with everywhere twice differentiable PDFs and then taking the limit.\end{remark}

We prove these theorems after developing the necessary machinery.

Let
\begin{eqnarray}
\label{equation M} M_{\mathcal{F}, n, \lambda}(G) & =& n\int_0^{\infty} x F(x)g(x) dx + \frac{\lambda}{n} \int_0^{\infty} \frac{f(x)}{1-G(x)}dx \\
 H_{\mathcal{F}, \mathcal{G},n, \lambda}(x) & =& nx F(x) + \frac{\lambda}n \int_0^{x} \frac{f(y)}{(1-G(y))^2} dy \\
\end{eqnarray}

For notational convenience, we will often omit the subscripts when their meaning is clear from the context (e.g. $M_{\mathcal{F}, n, \lambda}(G)$ and $H(x)_{\mathcal{F}, n, \lambda}(G)$ become
as $M(G)$ and  $H(x)$, respectively).

$M(G)$ will be a function that we would like to minimize, $\delta M_{|G}(\hat G-G)$ is the derivative at $\hat G-G$ in the direction $\hat G - \bar G$, which we will want to show is always $\geq 0$ at some local minimum.  $H(x)$ will turn out to be the derivative of $M(G)$, a sort of point-wise version of $\delta M_{|G}(\hat G)$.

%


\begin{claim} \label{claim derivative of M}  $$\delta M_{\mathcal{F}, n, \lambda|G}(\hat G - G) = \int_0^{\infty} (\hat g(x) - g(x)) H(x)dx$$
\end{claim}

\begin{proof}
This follows immediately from Claim~\ref{claim derivative of cost} and Claim~\ref{claim derivative of variance}.
%
%
\end{proof}

\begin{claim} \label{claim M is convex}  $M_{\mathcal{F}, n, \lambda}(G)$ is strictly convex. \end{claim}

\begin{proof}  By Claim~\ref{claim derivative of M} the first derivative of $M_{|G}$ is $H(x) = nx F(x) + \frac{\lambda}n \int_0^{x} \frac{f(x)}{(1-G(x))^2}$.  Taking the second derivatives we get $\frac{2 \lambda}n \int_0^{x} \frac{f(y)}{(1-G(y))^3}dy$, which is positive as long as we are above some point in the support of $\mathcal{F}$ (because by Claim~\ref{claim about G} we know that $G$ is between 0 and 1).  By our assumption, the only point less than every element in the support of $\mathcal{F}$ that we consider is $0$.  Thus, the derivative at $G$ in direction $\hat G$:  $$\int_0^{\infty} ( (\hat g(x) - g(x)) \frac{2 \lambda}n \int_0^{x} \frac{f(y)}{(1-G(y))^3}dy dx$$  is positive unless $G$ and the direction are the same except for at the point 0.  However, this is never the case if the direction is non-trivial.  And thus this is always strictly positive.
\end{proof}

\begin{lemma}  \label{lemma main M} For all $\hat G$, $\delta M_{\mathcal{F}, n, \lambda|\bar G_{\mathcal{F}, \frac{\lambda}{n^2}}}(\hat G - \bar G) \geq 0$.  Moreover, $\bar G_{\mathcal{F}, \frac{\lambda}{n^2}}$ is the unique CDF $G$ for which this holds.
\end{lemma}

\begin{proof}  By Claim~\ref{claim derivative of M} we see that we must show that for all $\hat{G}$, $\int_0^{\infty} (\hat g(x) - \bar g(x)) H(x) dx \geq 0$.  Both $\hat g$ and $\bar g$ are distribution density functions, so they integrate to one and their difference integrates to 0.   Necessarily off the support of $G$ we have $\hat g(x) \geq \bar g(x) =0$ .  Thus it is enough to show that $H(x)$ is constant on some region that contains the support of $G$ and that outside this region it is always greater than inside the region.

We show that \begin{enumerate}
 \item  $H$ is constant over the intervals  $[\bar x, \bar i_1]$, $[\bar j_k, \bar i_{k+1}]$ for $1 \leq i < m$,  and  $[\bar j_m, \infty)$; \\
 \item   For every $\bar I_k$, $H(\bar i_k) = H(\bar j_k)$ and for every $x \in \bar I_k$, $H(x) \geq H[j_k]$; and  \\
 \item  For $x \in [0, \bar x]$, $H(x) \geq H(\bar x)$.
\end{enumerate} These three steps suffice to show that for all $\hat G$, $\delta M_{\mathcal{F}, n, \lambda|\bar G_{\mathcal{F}, \frac{\lambda}{n^2}}}(\hat G - \bar G_{\mathcal{F}, \frac{\lambda}{n^2}}) \geq 0$.

Let $x \in  [\bar j_k, \bar i_{k+1}]$ for $1 \leq i < m$ then

\begin{align} H(x) - H(\bar j_k)   & =  n x F(x) - n \bar j_k F(\bar j_k) + \frac{\lambda}{n}  \int_{\bar j_k}^{x}  \frac{f(y)}{(1-\bar G_{\mathcal{F}, \frac{\lambda}{n^2}}(y))^2} dy  \\
                                   & =  n x F(x) - n \bar j_k F(\bar j_k) + \frac{\lambda}{n}  \int_{\bar j_k}^{x}  \frac{f(y)}{(1-\tilde G_{\mathcal{F}, \frac{\lambda}{n^2}}(y))^2} dy  \\
                                   & = 0
\end{align}

The last line follows because $\tilde{G}_{\mathcal{F}, \frac{\lambda}{n^2}}(x) = 1 - \sqrt{\frac{\lambda}{n^2} \frac{f(x)}{F(x) + xf(x)}}$ so $$\frac{\lambda}{n}  \int_{\tilde j_k}^{x}  \frac{f(y)}{(1- \tilde G_{\mathcal{F}, \frac{\lambda}{n^2}}(y))^2} dy = \frac{\lambda}{n}  \left( \frac{ n^2}{\lambda} y f(x) + F(y) \right)_{\tilde j_k}^{x} =   n x F(x) - n \bar j_k F(\bar j_k).$$  At the single point $\bar i_{k+1}$, it might not be that $\bar G_{\mathcal{F}, \frac{\lambda}{n^2}}(y) = \tilde G_{\mathcal{F}, \frac{\lambda}{n^2}}(y)$, but because this is a single point,  and $f(x)$ is differentiable, this will not affect the calculation.

The same argument works for $x \in [\bar x, \bar i_1]$, and $x \in [\bar j_m, \infty)$

To show 2) observe that:

\begin{align} H(\bar j_k) - H(x)   & =  n \bar j_k F(\bar j_k)- n x F(x)  + \frac{\lambda}{n}  \int_{x}^{\bar j_k}  \frac{f(y)}{(1-\bar G(y))^2} dy  \\
                                  & =n \bar j_k F(\bar j_k) - n x F(x) + \frac{\lambda}{n}  \int_{x}^{\bar j_k}  \frac{f(y)}{(\int_y^{\infty} \tilde g(z) dz)^2} dy
                                                + \frac{\lambda}{n} \left(\int_{x}^{\bar j_k}    \frac{f(y)}{(1-\bar G(y))^2} - \frac{f(y)}{(\int_y^{\infty} \tilde g(z) dz)^2} dy \right) \\
                  &=  \frac{\lambda}{n} \int_{x}^{\bar j_k}    \frac{f(y)}{(1-\bar G(y))^2} - \frac{f(y)}{(\int_y^{\infty} \tilde g(z) dz)^2} dy
\end{align}

However, by the definition of $\bar I_k$, for $\bar i_k \leq x \leq \bar j_k$ this is less that or equal to 0 with equality when $x = \bar i_k$.

Finally, we show that for $x \in [0, \bar x]$, $H(x) \leq H(\bar x)$.  As in the previous setup we see that:

\begin{align} H(\bar x) - H(x)  &=   \frac{\lambda}{n} \int_{x}^{\bar j_k}    \frac{f(y)}{(1-\bar G(y))^2} - \frac{f(y)}{(\int_y^{\infty} \tilde g(z) dz)^2} dy
\end{align}

But this is negative because for $x < \bar x$ we have that $\int_x^{\infty} \tilde g(z) dz \geq 1$, but $\bar G(y) = 0$.

The uniqueness follows directly from the fact that we found a local minimum and by Claim~\ref{claim M is convex} the function is strictly convex over our space which is also convex.
\end{proof}

\begin{proof}[of Theorem ~\ref{theorem threshold fixed n}]:
Let $F$ be twice differentiable for all positive $x$, and let $n$ be fixed.  We minimize $V^*(G)$ subject to a expected cost constraint $B$.  From Lagrange Multipliers (see section~\ref{subsec calc of var}) we have
\begin{align*} L[G, \lambda] & = \textrm{V}^*_{\mathcal{F}, n}(G) + \lambda \left(\textrm{Cost}_{\mathcal{F}, n}(\mathcal{G})- B\right)\\
                  & =  \frac{1}{n}\int_0^{\infty} \int_{0}^x \frac{f(y)}{1 - G(y)}dy dx + \lambda n \left(\int_0^{\infty} \left(x F(x) g(x) dx - B \right) \right)
\end{align*}

Thus we have that:
\begin{eqnarray}
 \label{eq1}\frac{\partial L[G, \lambda]}{\partial G}  &=  \delta M_{\mathcal{F}, n, \frac{1}{\lambda}}(G) & \geq 0 \\
 \label{eq2}\frac{\partial L[G, \lambda]}{\partial \lambda}  &= n \int_0^{\infty} x F(x) g(x) dx - B  &= 0
\end{eqnarray}

where $M$ is define in Equation~\ref{equation M}.
By Equation~\ref{eq2} we see that $\int_0^{\infty} n x F(x) \bar g_{\mathcal{F}, n,  \frac{1}{\lambda}}(x) dx = B/n$ and by Remark~\ref{remark cont in alpha} this uniquely determines $\lambda$.
From Lemma~\ref{lemma main M} we see that Equation~\ref{eq1} implies that $G = \bar G_{\mathcal{F}, \frac{1}{\lambda n^2}}$.

The proof is analogous when we are minimizing the cost with a bound on the variance.
\end{proof}

\begin{proof}[of Theorem ~\ref{theorem convex fixed n}]:
$$\delta \Phi_{| G}(\hat G - G) = \frac{\partial \Phi}{\partial x_1}_{|Cost(G)} \delta Cost(\hat G -G) + \frac{\partial \Phi}{\partial x_1}_{|V^*(G)} \delta V^*_{G}(\hat G -G) =
\delta M_{F, n,  \lambda|G}(\hat G - G)$$ where $n = \frac{\partial \Phi}{\partial x_1}_{|Cost(G)}$ and $\lambda = \frac{\partial \Phi}{\partial x_1}_{|V^*(G)} n$.
So for any fixed $\bar G_{\mathcal{F}, \alpha}$ to be optimal we must have $\delta M_{F, n,  \lambda|G}(\hat G - \bar G_{\mathcal{F}, \alpha}) \geq 0$ for all $\hat G$.
However, by Lemma \ref{lemma main M} we see that if this is the case, then $G$ must be of the form $\bar G_{\mathcal{F}, \alpha} (x)$ for some $$\alpha = \frac{\lambda}{n^2} = \frac{\frac{\partial \Phi}{\partial x_2}_{|V^*(\bar G_{\mathcal{F}, \alpha})}}{\frac{\partial \Phi}{\partial x_1}_{|Cost(\bar G_{\mathcal{F}, \alpha})}}. $$

The RHS, $\alpha$, is strictly monotonically increasing in $\alpha$ and goes from 0 to infinity as $\alpha$ does, where as the LHS is monotonically decreasing with respect to $\alpha$ (because  $V^*(\bar G_{\mathcal{F}, \alpha})$ decreases with $\alpha$ and $Cost^*(\bar G_{\mathcal{F}, \alpha})$ increases with $\alpha$, and $\Phi$ is convex), and is always positive because $\Phi$ is monotone.  Because all the functions are continuous, there exists a unique solution for $\alpha$.  Because $\Phi$, $Cost$ and $V^*$ are all convex the composition is, and thus this is a global maximum.
\end{proof}

Now we generalize the results to the setting where the number participants is not fixed a priori, but there may be some cost $\beta$ associated with each additional person recruited.

\begin{proof}[of Theorem ~\ref{theorem threshold variable n}]:

Thus the cost of the mechanism becomes $$\textrm{Cost}(\mathcal{F}, \mathcal{G}, n, \beta) = n \left( \beta + \int_0^{\infty} x F(x) g(x) dx\right)$$ and the variance remains the same.  We will deal with the setting where we would like to minimize the variance subject to some cost constraint.  This give us the following problem:

$$\textrm{Minimize    } \frac{1}{n}\int_{0}^\infty\frac{f(x)}{1-G(x)}$$
$$\textrm{Subject to:    } n \left( \beta +  \int_0^\infty c\cdot g(c)\cdot F(c) dc \right) = B$$
Using a Lagrange multiplier, we arrive at the equation:

\begin{align*} L[G, \lambda] & = \textrm{V}^*_{\mathcal{F}, n}(G) + \lambda \left(\textrm{Cost}_{\mathcal{F}, n}(\mathcal{G})- B\right)\\
                  & =  \frac{1}{n}\int_0^{\infty} \frac{f(x)}{1 - G(x)} dx + \lambda \left( n \left(\beta + \int_0^{\infty} \left(x F(x) g(x) dx \right) - B \right) \right)
\end{align*}

Thus we have that:
\begin{eqnarray}
 \label{eq1a}\frac{\partial L[G, n, \lambda]}{\partial G}  &=  \delta M_{\mathcal{F}, n, \frac{1}{\lambda}}(G) & \geq 0 \\
 \label{eq2a}\frac{\partial L[G, n, \lambda]}{\partial n}  &= \frac{-1}{n^2}\int_0^{\infty} \frac{f(x)}{1 - G(x)} dx + \lambda \left(\beta + \int_0^{\infty} x F(x) g(x) dx  \right)  &= 0\\
 \label{eq3a}\frac{\partial L[G, n, \lambda]}{\partial \lambda}  &= n \left(\beta +  \int_0^{\infty} x F(x) g(x) dx \right)- B  &= 0
\end{eqnarray}

where $M$ is define in Equation~\ref{equation M}.

From Lemma~\ref{lemma main M} we see that Equation~\ref{eq1a} implies that $G = \bar G_{\mathcal{F}, \frac{1}{\lambda n^2}}$.  If we let $\alpha = \frac{1}{\lambda n^2}$, then Equation~\ref{eq2a} then implies that
\begin{align} \label{eq-tradeoff}\sqrt{\alpha}(\textrm{V}^*_{\mathcal{F}, 1}( \bar G_{\mathcal{F}, \alpha})) = \frac{1}{\sqrt{\alpha}}(\textrm{Cost}_{\mathcal{F}, 1}( \bar G_{\mathcal{F}, \alpha}) + \beta). \end{align}

However,
\begin{align*}
\sqrt{\alpha}(\textrm{V}^*_{\mathcal{F}, 1}( \bar G_{\mathcal{F}, \alpha}))
  & = \sqrt{\alpha} \left[ \int_0^{\infty} \frac{f(x)}{\int_{x}^{\infty} \check g_{\mathcal{F}, \alpha}(y)dy} dx + \left(\int_0^{\bar x_{\mathcal{F}, \alpha}} \frac{f(x)}{\int_{x}^{\infty} \bar g_{\mathcal{F}, \alpha}(y) dy} - \frac{f(x)}{\int_{x}^{\infty} \check g_{\mathcal{F}, \alpha}(y) dy} dx \right)  \right]\\
  & = \int_0^{\infty} \frac{f(x)}{\int_{x}^{\infty} \check g_{\mathcal{F}, 1}(y) dy} dx +  \sqrt{\alpha} \left(\int_0^{\bar x_{\mathcal{F}, \alpha}} \frac{f(x)}{\int_{x}^{\infty} \bar g_{\mathcal{F}, \alpha}(y) dy} - \frac{f(x)}{\int_{x}^{\infty} \check g_{\mathcal{F}, \alpha}(y) dy} dx \right)
\end{align*}

And this increases in $\alpha$ because, by the definition of $\bar x_{\mathcal{F}, \alpha}$, the integrand in the second term is positive.  Similarly,

\begin{align*}
 \frac{1}{\sqrt{\alpha}}(\textrm{Cost}^*_{\mathcal{F}, 1}( \bar G_{\mathcal{F}, \alpha}) + \beta)
  & = \frac{1}{\sqrt{\alpha}} \left[ \int_0^{\infty} x F(x) \check g_{\mathcal{F}, \alpha}(x) dx - \left(\int_0^{\bar x_{\mathcal{F}, \alpha}} x F(x) \check g_{\mathcal{F}, \alpha}(x) dx \right) + \beta \right]\\
  & =  \int_0^{\infty} x F(x) \check g_{\mathcal{F}, 1}(x) dx -  \int_0^{\bar x_{\mathcal{F}, \alpha}} x F(x) \check g_{\mathcal{F}, 1}(x) + \frac{\beta}{\sqrt{\alpha}} dx
\end{align*}
And this decreases to 0 as $\alpha$ increases.  Moreover, if $\beta > 0$ then the RHS monotonically decreases from $+\infty$ to $0$, and so there exists a unique $\alpha$ that satisfies Equation~\ref{eq-tradeoff}.  Having derived a value for $\alpha$, Equation~\ref{eq3a} uniquely determines $n$.

In the other case, $\beta = 0$, and we can solve with different values of $n$, as in Theorem~\ref{theorem threshold fixed n}, and directly compare the values of the objective function.
As we have seen, the optimal $G$ is of the form $\bar G_{\mathcal{F}, \alpha}$ where $\alpha = \frac{\lambda}{n^2}$.  From Equation~\ref{eq2} (or Equation~\ref{eq2a}) and using a derivation similar to that above, we see that

$$\sqrt{\lambda} = B/\left(  \int_0^{\infty} x F(x) \check g_{\mathcal{F}, 1}(x) dx -  \int_0^{\bar x_{\mathcal{F}, \alpha}} x F(x) \check g_{\mathcal{F}, 1}(x) + \frac{\beta}{\sqrt{\alpha}} dx \right),$$ and the objective function becomes
\begin{align*}
\frac{1}{n}\int_{0}^\infty\frac{f(x)}{1-\bar G_{\mathcal{F}, \alpha}(x)} dx  &=   \sqrt{\lambda} \left( \int_0^{\infty} \frac{f(x)}{\int_{x}^{\infty} \check g_{\mathcal{F}, 1}(y) dy} dx +  \sqrt{\alpha} \left(\int_0^{\bar x_{\mathcal{F}, \alpha}} \frac{f(x)}{\int_{x}^{\infty} \bar g_{\mathcal{F}, \alpha}(y) dy} - \frac{f(x)}{\int_{x}^{\infty} \check g_{\mathcal{F}, \alpha}(y) dy} dx \right) \right) \\
& = B \frac{ \int_0^{\infty} \frac{f(x)}{\int_{x}^{\infty} \check g_{\mathcal{F}, 1}(y) dy} dx +  \sqrt{\alpha} \left(\int_0^{\bar x_{\mathcal{F}, \alpha}} \frac{f(x)}{\int_{x}^{\infty} \bar g_{\mathcal{F}, \alpha}(y) dy} - \frac{f(x)}{\int_{x}^{\infty} \check g_{\mathcal{F}, \alpha}(y) dy} dx \right) }
{ \int_0^{\infty} x F(x) \check g_{\mathcal{F}, 1}(x) dx -  \int_0^{\bar x_{\mathcal{F}, \alpha}} x F(x) \check g_{\mathcal{F}, 1}(x) + \frac{\beta}{\sqrt{\alpha}} dx}
\end{align*}
This is minimized if $\bar x_0 = 0$ because the extra term in the numerator is positive and the extra term in the denominator is negative. In the degenerate case where $\bar x_{\mathcal{F}, \alpha}$ is always greater than 0, a solution may not exist.  This happens if and only if there is mass on $\mathcal{F}$ arbitrarily close to 0.  The reason for this is that the surveyor will want an unbounded number of samples close to 0, because no matter how many samples he has, their price to variance trade off is cheap because their price is arbitrarily small.
Otherwise, after $n$ is above some threshold $\alpha$ will be sufficiently small so that $\bar x_0 = 0$.  At this point, the objective function will remain the same regardless of $n$, and there exists a inverse relationship between $\sqrt{\alpha}$ and $n$ because $\lambda$ is fixed and $\alpha = \frac{\lambda}{n^2}$.

%
%
\end{proof}

\section{Discussion and Open Problems}
In this paper, we have initiated the study of truthful, Bayesian-optimal survey design for generating unbiased estimators of a population statistic. Of course, there are many interesting directions left open for future work. For example:
\begin{enumerate}
\item In this paper, we have focused on mechanisms which promise \emph{unbiased} estimators, and have focused on optimizing properties of these mechanisms (i.e. variance, cost, etc.). However, if the surveyor is in the end interested in \emph{accuracy}, he may be willing to compute a \emph{biased} estimator if doing so can significantly reduce variance without increasing cost (error, in the end, is a function both of the bias of the estimator, and of its variance). What do optimal mechanisms look like in the regime of minimizing error, when we cannot necessarily restrict our attention to the Horvitz-Thompson estimator?
\item We have assumed a known prior over costs, but have otherwise made worst-case assumptions on what the distributions on surveyed types looks like. Is it possible to do substantially better in some cases if we make further (natural) assumptions on the distribution from which player types are drawn? One natural assumption might be that $\E[q(x)|c(x)]$ is monotonically increasing in $c(x)$. For example, we would expect this to be the case if $q(x) = 1$ represents some stigmatized property, and so such individuals might have higher cost for participating in such a survey.
\item What if costs are indeed drawn from some prior distribution, but this distribution is unknown to the mechanism. Can a mechanism which must now learn this distribution over costs compete with the optimal mechanisms derived here, with prior knowledge over the true distribution?
\end{enumerate}

\bibliographystyle{alpha}
\bibliography{truthfulsurveys}

\end{document}